\documentclass[letterpaper, 10 pt]{article}  

\usepackage{graphics} 
\usepackage{mathptmx} 
\usepackage{times} 

\usepackage{amsmath}
\usepackage{amsfonts}
\usepackage{multirow}
\usepackage{amssymb}
\usepackage{graphicx}
\usepackage{subfigure}
\usepackage{gensymb}
\usepackage[width=6in, height=9.00in]{geometry}
\usepackage{amsthm}
\usepackage{romannum}
\theoremstyle{definition}

\theoremstyle{remark}

\theoremstyle{plain}
\newtheorem{theorem}{Theorem}
\usepackage{multirow}
\usepackage{newtxtext,newtxmath}
\usepackage{booktabs}
\usepackage{adjustbox}
\usepackage{balance}

\usepackage{tabularx}
\newcolumntype{Y}{>{\centering\arraybackslash}X}
\makeatletter
\let\NAT@parse\undefined
\makeatother
\usepackage{hyperref} 
	\title{\LARGE \bf
		Modeling Human-Human Collaboration: A Connection Between Inter-Personal Motor Synergy and Consensus Algorithms
	}
	\author{Sara Honarvar,  Jin-OH Hahn, Tim Kiemel, Jae Kun Shim, and Yancy Diaz-Mercado
		\thanks{This work was supported by the Korea Institute of Machinery and Materials (KIMM) [NK210H, Development of core machinery technologies for autonomous operation and manufacturing].}
		\thanks{S. Honarvar, J.O. Hahn, and Y. Diaz-Mercado are with Department of Mechanical Engineering,
			University of Maryland, College Park, MD 20742 
			{\tt\small 
				honarvar@umd.edu, jhahn12@umd.edu, yancy@umd.edu}.}%
		\thanks{J.K. Shim and T. Kiemel are with Department of Kinesiology, 
			University of Maryland, College Park, MD 20742 
			{\tt\small jkshim@umd.edu, kiemel@umd.edu}.}%
	}
    \date{}	
	
\begin{document}

	\maketitle
 	\thispagestyle{empty}
	\pagestyle{empty}
	\begin{abstract}
		
		Many day-to-day activities involve people working collaboratively toward reaching a desired outcome. Previous research in motor control and neuroscience have proposed inter-personal motor synergy (IPMS) as a mechanism of collaboration between people, referring to the idea of how two or more people may work together "as if they were one" to coordinate their motion. In motor control literature, uncontrolled manifold (UCM) is used for quantifying IPMS. According to this approach, coordinated motion is achieved through stabilization of a performance variable (e.g., an output in a collaborative output tracking task). We show that the UCM approach is closely related to the well-studied consensus approach in multi-agent systems that concerns processes by which a set of interacting agents agree on a shared objective. To explore the connection between these two approaches, in this paper, we provide a control-theoretic model that represents the systems-level behaviors in a collaborative task. In particular, we utilize the consensus protocol and show how the model can be systematically tuned to reproduce the behavior exhibited by human-human collaboration experiments. We discuss the association between the proposed control law and the UCM approach and validate our model using experimental results previously collected from an inter-personal finger force production task. 
		
	\end{abstract}

	\section{Introduction}
	
	A common challenge in engineering, mathematics, and physical and biological sciences is how to address indeterminacy or redundancy encountered in many systems \cite{bernstein1966co}, oftentimes observed in overactuated systems. Humans have a complex body structure typically with more degrees-of-freedom (DOF) than those required for a particular motor task, and it is not clear how these mechanically/mathematically independent DOFs are controlled by the central nervous system (CNS) \cite{kim2020inter}. Even the seemingly simple task of grasping a glass of water requires the CNS to determine and synergistically control the forces and torques applied to the object by the thumb and fingers from infinite possibilities \cite{shim2005prehension1,shim2007prehension}. In human motor control research, this type of control is known as motor synergy, with two key features of error compensation and dimension compression \cite{scholz1999uncontrolled}. 
	
	In the motor control literature, motor synergies are defined as ``neural organization of a multi-element system that (1) organizes sharing of a task among a set of elemental variables; and (2) ensures task-specific co-variation among elemental variables with the purpose to stabilize performance variables" \cite{latash2007toward}. Consequently, if one of the elements deviates from its share of task such that it adversely affects the task goal, the other elements act in a manner to compensate for the introduced error in task performance. 
	
     The uncontrolled manifold (UCM) analysis \cite{scholz1999uncontrolled} provides a way to quantify synergies. More specifically, UCM analysis indicates which DOFs are potentially controlled by the CNS. Within this approach, for a given performance variable (e.g., total force or moment required for grasping a glass of water), the space of elemental variables (e.g., finger forces and moments produced in the grasping task) is decomposed into two orthogonal subspaces: i) the UCM subspace, in which variations of the elemental variables do not affect the performance variable; thus the CNS does not explicitly control variations within this space, and ii) the orthogonal complement to the UCM (ORT), where the CNS explicitly controls variations of elemental variables, as these variations change the performance variable \cite{scholz1999uncontrolled}. Motor synergy is typically quantified as the normalized difference between variability along these two subspaces \cite{latash2007toward}. Similarly, if there is not a clear performance variable, principal component analysis (PCA) is used for examining the variability in performance \cite{ramenzoni2012interpersonal}. PCA finds relationships in high-dimensional datasets and maps them into a space with principal components that reflect the dataset's major dimensions of variation \cite{daffertshofer2004pca}. 
	
	When a person works with others to lift furniture, for example, individual CNSs are required to work together in order to produce a desired behavioral outcome; in other words, the motor system behaves cooperatively to achieve a common motor task. Previous studies showed inter-personal motor synergy (IPMS), with key properties of error compensation and dimension reduction, as the potential mechanism responsible for this cooperation among independent central nervous systems (CNSs) in performing shared motor tasks \cite{riley2011interpersonal,HonarvarIPMS}. The UCM and PCA approaches have also been used for characterizing IPMS \cite{riley2011interpersonal,black2007synergies}.
	
	The past decade has seen remarkable advances made in our understanding of motor synergies within a person \cite{shim2007prehension,honarvar11unveiling}. In contrast, our understanding of IPMS is extremely limited \cite{riley2011interpersonal} though many of the tasks that we do in our everyday life involve multiple people working together, such as lifting and carrying a piece of furniture together. Co-working is not just common among people, but in many other settings; such as a team of robots that coordinate to achieve a desired task \cite{notomista2019constraint}. In the field of control theory and robotics, decentralized cooperation in multi-agent networked systems have been studied through the framework of consensus algorithms \cite{olfati2007consensus}. Consensus algorithms allow for a team of agents to achieve agreement on a common objective in a decentralized fashion by only interacting with neighboring agents. Here, we aim to establish a connection between consensus algorithms and IPMS, and use this relationship to model human-human collaboration. The result is an intuitive control-theoretic model that captures how two or multiple people can achieve a common objective in a decentralized fashion. 
	
	Control-theoretic neural mechanism models that can reproduce sensorimotor interactions and motor synergies in intra-personal tasks exist \cite{honarvar11unveiling}. However, mathematical models that can elucidate the relationship between sensory integration and motor synergies, especially in the context of inter-personal motor redundancy and synergies, are not yet available to the best of our knowledge. The approach taken in this paper is inspired by prior behavioral performance observed in a human-human collaboration (HHC) experiment \cite{thesis}. To elucidate the mechanisms responsible for IPMS, we provide a control-theoretic model that can reproduce ensemble-level behaviors of multiple CNSs working together. 
	
	The rest of this paper is organized as follows: In Section II, we provide a description of the HHC experiment, and the mathematical characterization of the problem under consideration. In Section III, we provide a two-level, namely task-level and inter-agent level, control strategy for collaborative output-tracking based on Lyapunov analysis and consenus algorithms.  In Section IV, we provide the results of our simulation and compare it against the HHC experimental data. Section V provides conclusions on our findings.
	
	\section{Preliminaries}
	In this section, we provide a short description of the HHC experiment that motivates the model, present the mathematical formulation of the proposed finger model, and then state the collaborative output-tracking problem.
	\subsection{HHC Experiment}
	The model provided in this paper is inspired by a previously conducted HHC experiment \cite{thesis}. In this experiment, 21 pairs of people (i.e., 42 participants) were asked to work together and perform a finger force production task by pressing on force sensors placed at two ends of a seesaw-like apparatus. A monitor was placed in front of each co-worker to provide them with visual feedback on the task. Co-workers were required to reach a target force and maintain it for 23 seconds over 10 trials. They received visual feedback of the target force (i.e., $5N$) and their averaged combined force. They were asked to remain quiet during the experiment and not communicate verbally. The force profile of index and middle fingers of each participant were recorded during the experiment. Full details of the experiment can be found in \cite{thesis}. 
	
	A second-order Butterworth low pass filter with 15 Hz as the cut-off frequency was used to eliminate the force sensor measurement noise \cite{bumannusing}. From the experiment, we had 210 combination of data for all subjects and all trials (i.e., 21 pairs by 10 trials of finger force trajectories). Removing outliers (i.e., trials with finger forces less than $1N$, or 20$\%$ of the average), provided us with 182 trials of force data (Fig.~\ref{fig1}).
	
	The experimental data showed that all the pairs could manage to reach the target force and maintain this objective with minimal error \cite{thesis}. Without verbal communication and only relying on visual feedback, co-workers formed IPMS; they minimized the error between the target force and the combined output by compensating for each other’s error. In other words, they collaborated synergistically at inter-personal level to satisfy the goal at the task level \cite{thesis}.
	
	
	\begin{figure*}[t!]
		\centering
		\includegraphics[width=0.495\linewidth,keepaspectratio]{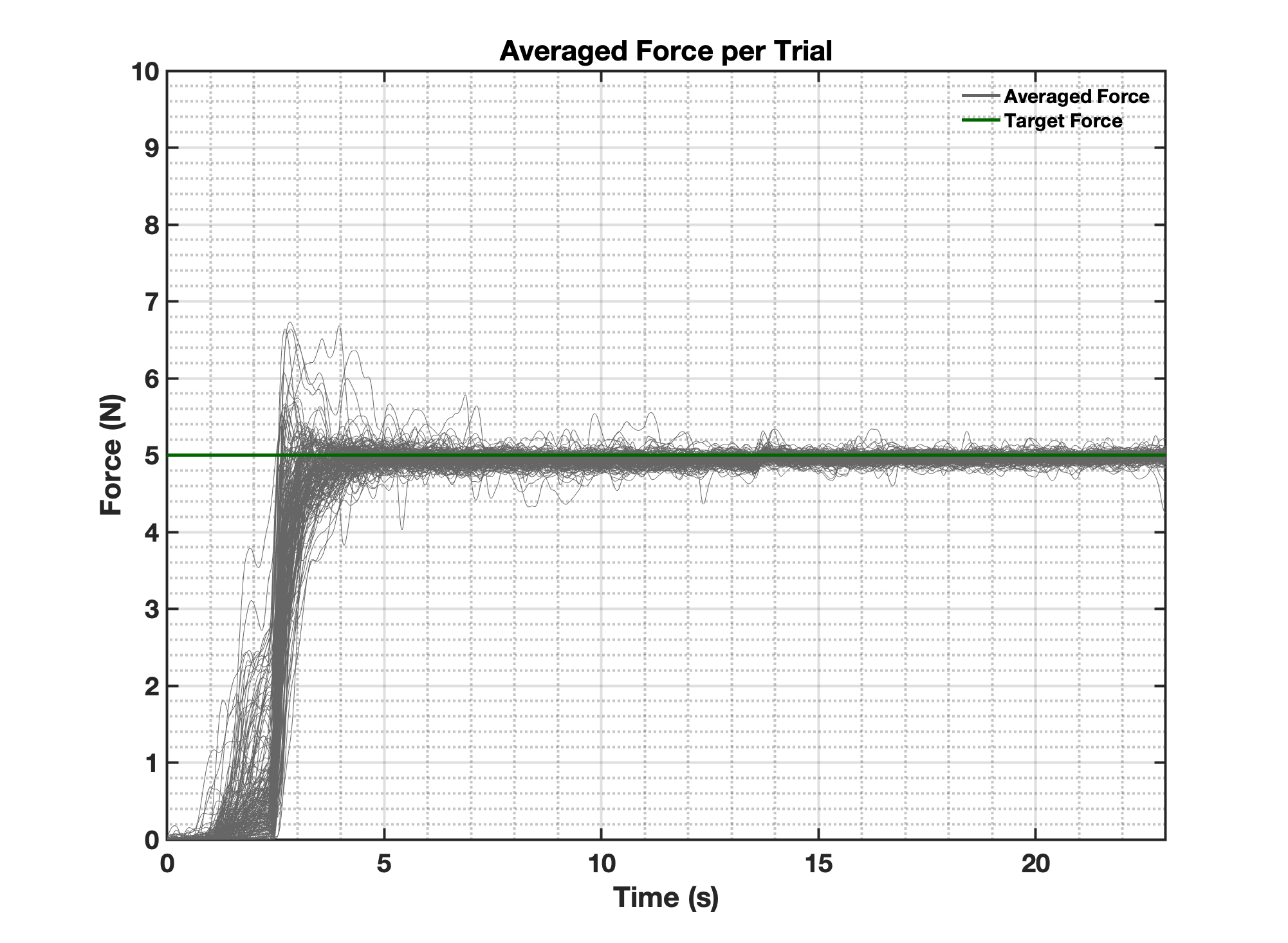}
			\includegraphics[width=0.495\linewidth,keepaspectratio]{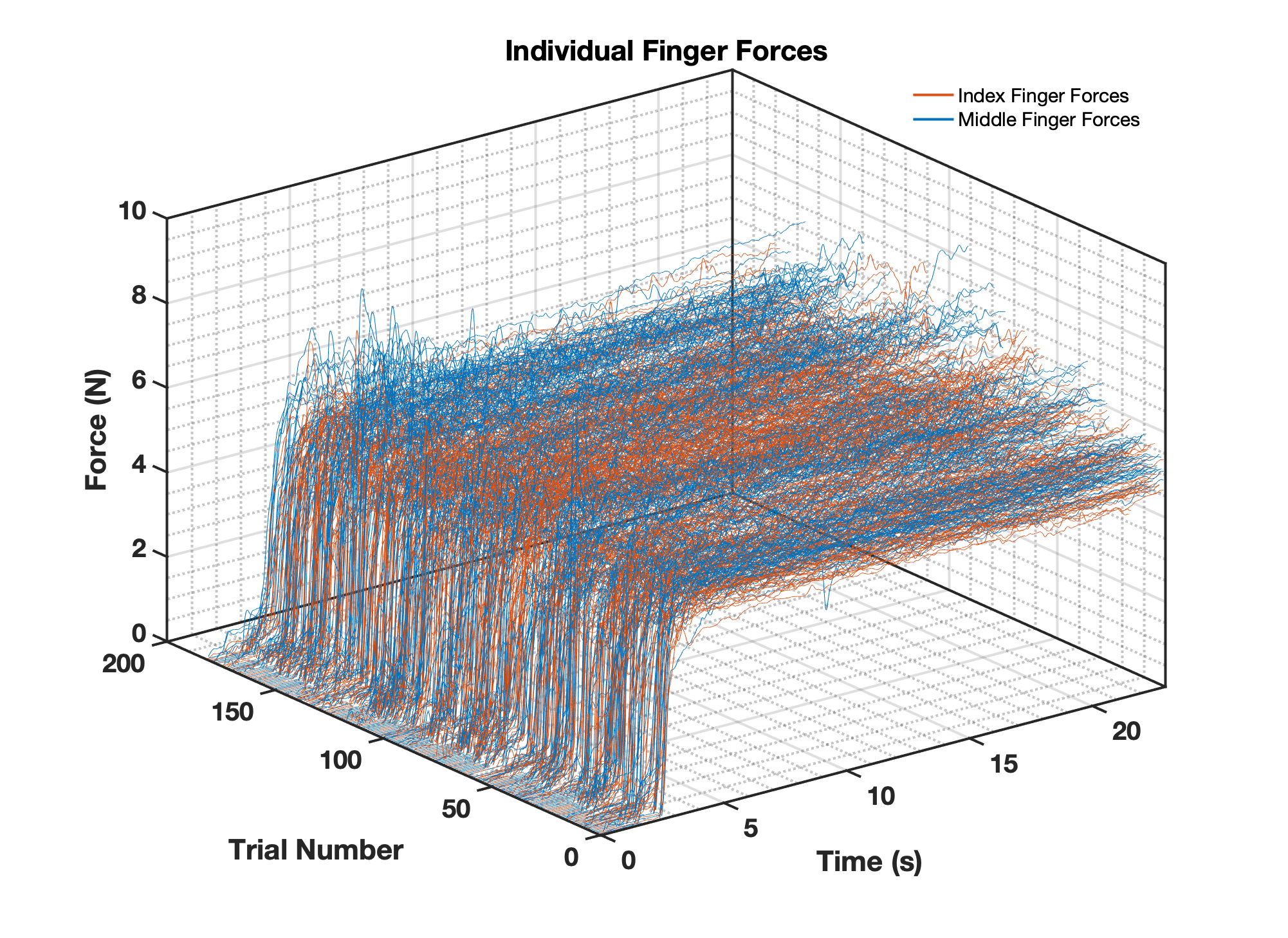}
		\caption{Demonstration of experimental co-working data. The left plot presents the average force generated per trial (182) by the four individual fingers. The right plot shows the individual finger data per trial. The orange corresponds to forces generated by index fingers, and the blue to middle finger.}
		\label{fig1}
	\vspace{-10pt}	
	\end{figure*}
	
	%
	\subsection{Topology of Interaction}
	We use notations from algebraic graph theory \cite{mesbahi2010graph} to represent the interaction between co-working agents. A graph, $G = (V,E)$, consists of a set of nodes $V = \{1,\dots,N\}$, and a set of edges $E \subset V \times V$, representing the interaction between agents. The adjacency matrix is defined such that its $ij$-th entry is given by $[A]_{ij}=a_{ij}$, where $a_{ij}=1$ if $(i,j)\in E$ and $a_{ij}=0$ otherwise. We will interpret the $(i,j)\in E$ to mean that certain information about the $i$th agent is available to the $j$th agent. Let $\mathcal{N}_i =\{j\in V |(j,i)\in E\}$ denote the set of neighbors for the $i$th agent. In this paper, we assume edges between  agents to be bidirectional and focus our analysis on undirected graphs. An undirected graph is said to be connected if there is a path between any two agents.
	The entries of the graph Laplacian are defined as $[L]_{ij}=l_{ij}$, where
	\begin{equation}
	l_{ij} = \begin{cases} |\mathcal{N}_i|, & j=i\\
	-a_{ij}, & j\neq i
	\end{cases}
	\end{equation}
	and $|\mathcal{N}_i|$ denotes cardinality of agent $i$'s neighbor set. The graph Laplacian has several useful spectral properties \cite{mesbahi2010graph}. For an undirected graph, $L$ is a symmetric positive semi-definite matrix with at least one zero eigenvalue ($\lambda_1=0$) with corresponding eigenvector $\mathbf1_N = [1 \dots 1]^T$, satisfying $L\mathbf1_N=0$. Further, when $G$ is connected, $L$ has exactly one zero eigenvalue (with eigenvector $\mathbf{1}_N$).
	\subsection{Level of Error Compensation}
	The level of error compensation between agents, one of the features of motor synergies, can be captured by the sharing ratio, more specifically how the CNS shares the task among agents \cite{honarvar11unveiling}. To capture the level of error compensation, we sample the share of the forces generated by each of the agents, $s_i$, for $i=1,\ldots,4$, from a normalized, truncated multivariate normal distribution with mean of $1/n$, and a $n \times n$ covariance, $C_m$, with $n=4$, such that $s_i\in[0,1]$ and $\sum_is_i=1$. To estimate $C_m$, we calculated the covariance matrix, $C_i$, of individual finger forces, $y_i$, for all $nTrials=182$ combinations of forces (in steady state) and then get the average of covariance across all combinations, where 
	
	\begin{align}\label{cov}
	\begin{split}
	{C}_i &= 
	E[(y_i-E[y_i])(y_i-E[y_i])^T], \text{ for }i=1,\ldots, 4\\
	{C}_m &= \frac{1}{T-1}\sum_i^{nTrials} {C}_i
	\end{split}
 	\end{align}
and the operator $E[\cdot]$ denotes the expected value (mean) of its argument.
	\subsection{Agent Dynamics}
We consider a network of double integrators to represent finger models contaminated with sensory and motor noises \cite{honarvar11unveiling} for a collection of $N$ agents. The dynamics of the $i$th agent in normal form is as follows:
	\begin{equation}\label{eq:finger_dynamics}
	\begin{aligned}
	\dot{z}_i &= A_c {z}_i + B_c {u}_i + {w}_i\\
	{y}_i &= C_c{z}_i + {v}_i  
	\end{aligned}
	\end{equation}
where $ z_i = [y_i,\dot{y}_i]^T\in\mathbb{R}^2$ is the agent $i$'s state (including the force, $y_i$, and rate of force, $\dot{y}_i$), and $u_i \in\mathbb{R} $, is the control input of agent $i$, for $i=1,\dots, N$. The process noise and measurement noises are represented by Gaussian random vectors $w_i \sim \mathcal{N}(0,Q)$ and $v_i \sim \mathcal{N}(0,R)$ with zero mean and covariances $Q$ and $R$, respectively. The noise vectors are all assumed to be mutually independent. Finally, $ (A_c,B_c,C_c) $ represent a chain of integrators system, i.e.,
	\begin{align*}
	A_c &= 
	\begin{bmatrix}
	0      & 1     \\
	0      & 0           
	\end{bmatrix}, 
	&
	B_c &= 
	\begin{bmatrix}
	0      \\
	1
	\end{bmatrix},
	&
	C_c &= 
	\begin{bmatrix}
	1 & 0
	\end{bmatrix}.		
	\end{align*}
 It is presumed that each agent is aware of its own states, but not necessarily that of others.	
 \subsection{State Estimation Using Kalman Filter}
The standard continuous time Kalman filter \cite{lewis2017optimal} is employed to reduce the uncertainty in the state estimate of each agent, perturbed by sensory and motor noises (\cite{wolpert2000computational}). 
Assuming $P(0)=P_0$, $\hat{z}(0)=\bar{z}_0$, the error covariance update is given by
 \begin{equation}
 \dot{P}=A_cP+PA_c^T+Q-PC_c^TR^{-1}C_cP.
 \end{equation}
 So, the state estimate update is given by
 \begin{equation}
 \dot{\hat{z}}=A_c\hat{z}+B_cu+K(y-C_c\hat{z}).
 \end{equation}
 where $K=PC_c^TR^{-1}$ is the Kalman gain.
 
	
	
	\subsection{Desired Averaged Dynamics}
	In the HHC experiments, agents were tasked with performing collaborative output-tracking of the target output $ y_t $. Agents were only provided with the averaged output contributed by the team, i.e.,
	\begin{align}\label{eq:overallOutput}
	\hat{\bar{{y}}} = \frac{\hat{y}_o}{N}=\frac{\mathbf{1}_N^T}{N}\hat{y},
	\end{align}
	and (presumably) a notion of their own individual output contributions, where $\hat{y}_o=\sum_{i=1}^N \hat{y}_i$ is the overall output contributed by the team and $\hat{y}=[\hat{y}_1,\hat{y}_2,...,\hat{y}_N]^T$ is the individual agents' output.
	It follows that the $ k $\textsuperscript{th} derivative of the average output is given by
	\begin{align}\label{eq:overallOutputDerivative}
	\hat{\bar{{y}}}^{(k)} = \frac{\hat{y}_o^{(k)}}{N}=\frac{\mathbf{1}_N^T}{N}\hat{y}^{(k)}. 
	\end{align}
	The dynamics of $\hat{\bar{y}}$ can be represented as a chain of integrators with $\bar{z} = [\hat{\bar{y}},\dot{\hat{\bar{y}}}]^T$ via
	\begin{align}\label{eq:zoStateSpace}
	\dot{\bar{z}} = A_c \bar{z} + B_c \hat{\ddot{\bar{{y}}}}.
	\end{align}
	Typical performances of the overall output force contributions from the HHC experiment (Fig.~\ref{fig1}) demonstrate a dominant second order behavior for the overall output. Thus, in order to model the experimental data, we wish to obtain a desired output response ($\hat{\bar{y}}_d$) for the average output as dictated by the second order desired system dynamics. 
	\begin{equation}\label{eq:desiredDynamics}
	\begin{aligned}
	\dot{\bar{z}}_d &= A_d \bar{z}_d + B_d {y}_{t}\\
	\hat{\bar{y}}_d &= C_d \bar{z}_d.
	\end{aligned}
	\end{equation}
	with state $ \bar{z}_d = [\hat{\bar{{y}}}_d,\hat{\dot{\bar{{y}}}}_d]^T $ and the individual finger's target force $y_t=5N$. Besides, the system $ (A_d,B_d,C_d) $ is assumed to be linear time-invariant and stable in accordance with experimental observations in \cite{thesis}.
	The desired system dynamics can also be written in normal form \cite{khalil2002nonlinear}, namely 
	\begin{equation*}
	\dot{\bar{z}}_d = A_c \bar{z}_d + B_c\hat{\ddot{\bar{y}}}_d
	\end{equation*}
	where $\hat{\ddot{\bar{y}}}_d= \alpha_d\bar{z}_d + \gamma_d {y}_t$, $\alpha_d(\bar{z}_d)=C_dA_d^2$ and $\gamma_d = C_dA_dB_d$.
	\section{Control Law}
	The objective for the HHC experiment was to collaboratively achieve and maintain a target output. We propose that the following node-level control law which solely relies on the information available to the subjects, namely the individual's states and the average output states can achieve the objective of HHC experiment.
	\begin{equation}\label{eq:controlLawInTermsOfAverage}
	u_i=\hat{\ddot{\bar{{y}}}} + (\hat{\bar{{y}}}-y_i) + \eta(\hat{\dot{\bar{y}}}-\hat{\dot{y}}_i) - y_t ( \tfrac{1}{N} - s_i )
	\end{equation}
	where $\eta>0$ is a tuning parameter, and $\hat{\ddot{\bar{{y}}}}$ will be discussed shortly. 
	
	Inspired by the UCM hypothesis \cite{scholz1999uncontrolled}, the system dynamics in \eqref{eq:controlLawInTermsOfAverage} can be decomposed into two orthogonal components: i) the $\operatorname{span}\{\mathbf{1}_N\}$ components, where changes lead to deviation from the task goal (similar to ORT subspace in UCM approach) and ii) the components orthogonal to $\operatorname{span}\{\mathbf{1}_N\}$, where changes do not lead to deviations from the task goal and do not need to be explicitly controlled to achieve the task (similar to UCM subspace in UCM approach). Correspondingly, we consider two components for the ensemble-level proposed control law, $u = [u_1,\ldots,u_N]^T= u^{\parallel} + u^{\perp}$, where: i) the task-level component, $u^{\parallel}$, which ensures the asymptotic convergence of the average output to the desired dynamics, and ii) the inter-agent level component, $u^{\perp}$, that we will show can be represented using consensus algorithms to represent the collaboration between agents to agree on a common objective. 
	\subsection{Ensemble-Level Dynamics}
	
	The control law in \eqref{eq:controlLawInTermsOfAverage} can be manipulated to more clearly show how an agent is affected by the individual agents through the average output. In particular, note that we may rewrite the equation as
	\begin{equation}\label{eq:controlLawInConsensusForm}
	u_i = \hat{\ddot{\bar{y}}}+\tfrac{1}{N}\sum_{j=1}^N\left[(\hat{y}_j-\hat{y}_i)+\eta(\hat{\dot{y}}_j-\hat{\dot{y}}_i) - y_t (s_j-s_i)\right].
	\end{equation}
	Stacking these into a single vector, the ensemble-level control law can be expressed as
	\begin{equation}\label{eq:controlLawEnsembleForm}
	u = \hat{\ddot{\bar{y}}}\mathbf{1}_N+\tfrac{1}{N} \left(\hat{y}_tLs - L\hat{y}-\eta L\hat{\dot{y}}\right)
	\end{equation}
	where $s=[s_1,\ldots,s_N]^T$, and $L$ is the graph Laplacian associated with the complete graph, i.e., where for every agent $\mathcal{N}_i = \{j\in V ~|~ j\neq i\}$. From this, and due to the fact that for an undirected graph $\mathbf{1}_N^TL=0$, we get that
	\begin{equation}\label{eq:controlLawComponents}
	\begin{split}
	u^\parallel &=\tfrac{1}{N}\mathbf{1}_N\mathbf{1}_N^Tu =  \hat{\ddot{\bar{y}}}\mathbf{1}_N,\\
	 u^\perp &= (I-\tfrac{1}{N}\mathbf{1}_N\mathbf{1}_N^T)u \\&= \tfrac{1}{N}\sum_{j=1}^N\left[(\hat{y}_j-\hat{y}_i)+\eta(\hat{\dot{y}}_j-\hat{\dot{y}}_i) - y_t (s_j-s_i)\right].
	\end{split}
	\end{equation}
	where we use the property that $(I-\tfrac{1}{N}\mathbf{1}_N\mathbf{1}_N^T)L = L$ \cite{mesbahi2010graph}.
	\subsection{Task-Level Convergence Analysis}
	From \eqref{eq:controlLawComponents} it follows that only the term $\hat{\ddot{\bar{y}}}$ affects the task-level component of the agents dynamics. Thus, we will design this term so that the averaged trajectories converge to that of the desired dynamics. This is captured by the following theorem. 
	\begin{theorem}\label{overal}
		For the agents with dynamics as in \eqref{eq:finger_dynamics}, the averaged output of \eqref{eq:overallOutput} asymptotically approaches the output of \eqref{eq:desiredDynamics} if $ \hat{\ddot{\bar{y}}} =  \alpha_d\bar{z} + \gamma_d {y_t}- \tfrac{1}{2} B_c^TP_c(\bar{z}-\bar{z}_d) $ where $ P_c $ is the solution to
		\begin{align*}
		P_c(A_c+B_c\alpha_d)+(A_c+B_c\alpha_d)^TP_c\\-P_cB_c^{\vphantom{T}}B_c^TP_c + Q = 0
		\end{align*}
		for $ Q=Q^T\succ0 $, with $ A_c+B_c \alpha_d $ stable and $(Q, A_c  + B_c \alpha_d)$ observable.
		
	\end{theorem}
	\begin{proof}
		Using the Lyapunov candidate
		\begin{align}\label{eq:LyapunovCandidate}
		V = (\bar{z}-\bar{z}_d)^TP_c(\bar{z}-\bar{z}_d)
		\end{align}
		for some $ P_c=P_c^T\succ0 $, it can be shown that
		the trajectories of $\bar{z}$ approach those of $\bar{z}_d$ asymptotically by the Lyapunov stability theorem \cite{khalil2002nonlinear}.
	\end{proof}
	 As a result, $u^{\parallel}$ ensures the convergence of the $\operatorname{span}\{\mathbf{1_N}\}$ components of the dynamics. Now that we know the averaged dynamics converge, we show the behavior of the UCM dynamics at inter-agent level.
	\subsection{UCM Convergence Analysis}
	In the previous section, we showed that the error dynamics between $\bar{z}$ and $\bar{z}_d$ vanish after some time. Now, we demonstrate that the consensus algorithm can achieve the desired level of error compensation exhibited by the individuals in the HHC experiments as captured by the sharing of the task among agents ($s_i$).
%
%
	Let $\Psi^T=[\Psi_1^T, \Psi_2^T]$, where, $\Psi_1 =\frac{1}{N}[y_1-s_1y_t,\dots,y_N-s_Ny_t]^T$, $\Psi_2 = \dot{\Psi}_{1}=\frac{1}{N}[\dot{y}_1,\dot{y}_2,...,\dot{y}_N]^T$.  We thus have that
	\begin{equation}\label{eq:Psi-consensus}
	\dot{\Psi}=\begin{bmatrix}
	\dot{\Psi}_{1}\\\dot{\Psi}_{2}
	\end{bmatrix}=\begin{bmatrix}
	0_N & I_N\\-L &-\eta L
	\end{bmatrix}\begin{bmatrix}
	\Psi_{1}\\ \Psi_{2}
	\end{bmatrix}.
	\end{equation}
	In other words, if consensus is achieved then $\Psi(t)\to T$ as $t\to\infty$, where 
	\begin{equation*}
	T = \begin{Bmatrix}
	\Psi\in \mathbb{R}^{2N}|\Psi_{1}\in \operatorname{span}\{\mathbf{1_N}\}, \Psi_{2}\in \operatorname{span}\{\mathbf{1_N}\}
	\end{Bmatrix}.
	\end{equation*}
	The convergence analysis of (\ref{eq:Psi-consensus}) can be performed through Lyapunov analysis similar to Theorem 1 of \cite{qin2011coordination}. This is captured by the following Theorem.
	
	\begin{theorem}\label{convergence_of_consensus}
		For any $\eta>0$, the UCM component of the control law in \eqref{eq:controlLawComponents}, $u^{\perp}$, will lead to consensus on the deviation from the share of the average force, i.e., $\Psi$ will asymptotically converge to the set $T$.
	\end{theorem}
	\begin{proof}
	    Proof follows from the choice of the positive definite Lyapunov function candidate:
		\begin{equation}\label{lyapunov-consensus}
		V = \frac{1}{2}\delta^TL\delta+\frac{1}{2}\|\dot{\delta}\|^2
		\end{equation}
		where $\delta =\Pi \Psi_{1}$ is the disagreement vector, and $\Pi = (I-(1/N)\mathbf{1_N}\mathbf{1_N}^T)$ is the projection matrix onto the orthogonal complement to $\operatorname{span}\{\mathbf{1_N}\}$.

	\end{proof}
	It should be noted that Theorem \ref{convergence_of_consensus} only suggests that $\dot{y}_i$ will tend to $\operatorname{span}\{\mathbf{1_N}\}$, however, due to the convergence of the average dynamics under Theorem 1, for static reference $y_t$, $\hat{\dot{y}}_i$ will in fact tend to 0. Using Theorem \ref{overal} and \ref{convergence_of_consensus}, we showed that $\Pi y\to\Pi y_ts$ and $\hat{\bar{y}}\to\hat{\bar{y}}_d$ under the choice of $u_i$ in equation \eqref{eq:controlLawComponents}, as desired.
	
	
	\begin{figure}[t!]
		\centering
			\includegraphics[width=\linewidth]{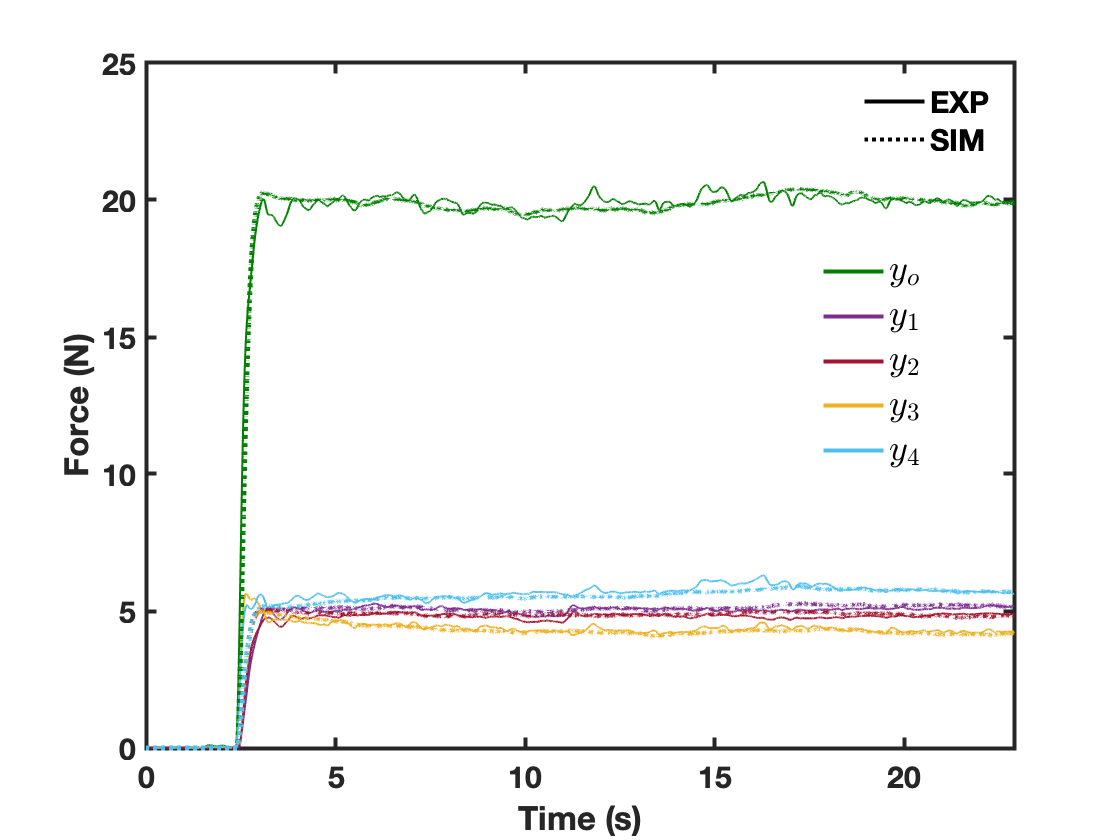}
		\caption{Comparison between simulation (dashed line) and experiment for a representative trial.} 
		\label{fig2}
		\vspace{-10pt}
	\end{figure}
	
	

	\section{Simulation Results}
	In this section, we present the result of our simulations using the proposed control law and compare them with the HHC experimental result \cite{thesis} to validate our model.
	\subsection{Choice of Parameters}
The desired averaged dynamics were designed as a second-order system in standard form (Fig.~\ref{fig1}). To define the desired behavior, we estimated the desired damping ratio, $\zeta$, and natural frequency, $\omega_n$, empirically for each 182 trials (e.g., $\zeta =0.72 \pm 0.1$, and $\omega_n = 9.32 \pm 2.01$). We used the mean values $\zeta =0.72$ and $\omega_n = 9.32$ as nominal values for simulation, which resulted in
 \begin{align*}
	A_d &= 
	\begin{bmatrix}
	0      & 1     \\
	-82      & -13.04           
	\end{bmatrix}, 
	&
	B_d &= 
	\begin{bmatrix}
	0      \\
	82.004
	\end{bmatrix},
	&
	C_d &= 
	\begin{bmatrix}
	1 & 0
	\end{bmatrix}.		
	\end{align*}
 The $\eta =1/\tau$ is a positive gain for consensus on the rate of force production, where $\tau$ is the reaction time. In motor control studies, it has been shown that this rate is the key response parameter determining the reaction time in human after receiving visual feedback for producing finger force production tasks \cite{carlton1987reaction}. The reaction time between the presentation of visual information and the initiation of a movement correction for finger force production has been reported between 100-180 ms, with the average value of 135 msec \cite{carlsen2009differential}. For this reason, we used the nominal value of $\eta = 7.41$ for our simulation results.
	   
The sharing ratios, were sampled from a normalized, truncated multivariate normal distribution as $s_i \sim \mathcal{N}(1/4,C_m)$, for $i=1,\ldots4$, and the covariance,
\begin{equation*}
 C_m = \begin{bmatrix}
	0.0242  &  0.0038  & -0.0069&   -0.0106\\
    0.0038  &  0.0217 &  -0.0072 &  -0.0087\\
   -0.0069  & -0.0072   & 0.0270 &  -0.0015\\
   -0.0106  & -0.0087  & -0.0015 &   0.0335
	\end{bmatrix}.   
\end{equation*}
As mentioned earlier, we calculated the covariance between individual agents' forces in steady state (i.e., 16-23s) within each trial and then calculated the average of covariances across 182 trials to get $C_m$.

The measurement noise variance, $R = 2.24$, was estimated from the average across 182 trials of the experimental force variance in the transient section (i.e., 2-16s). The process noise covariance $Q= 9.5\times 10^{-2}\, I_{2x2}$ was selected based on \cite{wolpert1995internal}. 

	
	
	
	\subsection{Model Validation}
	We examined the validity of the model from different aspects: First, we show that the simulation results accomplished the task goal of attaining and sustaining the target force similar to the HHC experiment. For that, we examined the root mean square error (RMSE) between the overall output and the target force for all trials in both simulation and experiment in the steady state. We found that the mean RMSE $\pm$ standard error across trials for simulation ($0.217 \pm 0.088$) closely matched that of the experiment ($0.246 \pm 0.094$).  
	Furthermore, the proposed model is capable of replicating individual trials as well as the force variability across trials.
	
	\subsubsection{Replicating a Trial}

    With the right choice of parameters for the desired behavior and the sharing ratio, the model can replicate individual trials. As an example, in Fig.~\ref{fig2}, it can be seen that the simulation results for a particular trial closely match the combined and individual agent's outputs of the experimental result. The parameters for this sample trial were: desired behavior: $\zeta=0.80$, ${\omega}_n=7.86$, $\eta=6$, and the sharing ratio $s^T=\begin{bmatrix}
    0.2562,\,0.2458,\, 0.2118,\, 0.2861
    \end{bmatrix}$, all the other parameters including the measurement and process noises remained intact. Given the randomness caused by the measurement and process noise, we ran the simulation with these particular choice of parameters for 5 times and got the averaged of the 5 simulations. RMSE (e.g., mean $\pm$ standard error) between the sample experiment and the averaged simulation results are summarized in Table ~\ref{tab:sample}. 
	
	\subsubsection{Examining Force Variability Using PCA and UCM}
	
	We calculated PCA for individual agents' output in steady state (16-23s). As shown in Table~\ref{tab:PCA}, we found that for both experiment and simulation, about 90\% of variance is along the UCM, defined by principal components (PC) 1 through 3 (see explained variance in Table~\ref{tab:PCA}). These PC's are nearly orthogonal to $\operatorname{span}\{\mathbf{1_N}\}$ (ORT), as shown in Table~\ref{tab:PCA} by evaluating their angles with $\operatorname{span}\{\mathbf{1_N}\}$. This finding suggests that the task level component of the control law ensures the convergence of $\operatorname{span}\{\mathbf{1_N}\}$ components of the dynamics and is an indication of dimension compression of motor synergies which is replicated by our model.
	
	To further validate the results of our simulation, we compared the experiment and simulation in terms of inter-agent synergy through the framework of UCM \cite{latash2007toward}. We computed the UCM and ORT basis for the steady state data (i.e., 16-23s) and project the data onto those basis. Finally, we computed the variance along UCM ($V_{UCM}$) and ORT ($V_{ORT}$) as well as synergy index (i.e., $\Delta V$) for both experiment and simulation (Table~\ref{tab:UCM}). Synergy index \cite{latash2007toward} was defined as $\Delta V = \frac{V_{UCM}/3 - V_{ORT}}{V_{TOT}/4}.$ As shown in Table~\ref{tab:UCM}, the proposed control law can reproduce the synergistic interaction between agents. 
	
	\begin{table}[!tbp]
		\centering
		\caption{The Level of RMSE (Mean $\pm$ Standard Error) Between Simulation and Experimental Forces For a Representative Trial. }
			\begin{tabularx}{0.5\linewidth}{|c|Y|}
				\hline
				Force & RMSE (N)\\ [0.5ex] 
				\hline\hline
				$e_{{y}_{o}}$ & 0.43 $\pm$ 0.02\\	[0.3ex] 
				\hline
				$e_{{y}_{1}}$ & 0.22 $\pm$ 0.01\\[0.3ex] 
				\hline
				$e_{{y}_{2}}$ & 0.22 $\pm$ 0.04\\[0.3ex] 
				\hline
				$e_{{y}_{3}}$ & 0.32 $\pm$ 0.01\\[0.3ex] 
				\hline
				$e_{{y}_{4}}$ & 0.33 $\pm$ 0.03\\[0.3ex] 
				\hline		
			\end{tabularx}
		\label{tab:sample}
	\end{table}
	
    \begin{table}[!tbp]
        \centering
		\caption{Summary of PCA results. The explained variance (i.e., the percentage of the total variance explained by each PC), and the angle of each PC with $\operatorname{span}\{\mathbf{1_N}\}$ are reported. }	\label{tab:PCA}
		\setlength{\tabcolsep}{2pt} 
        \renewcommand{\arraystretch}{1} 
			\begin{tabularx}{\linewidth}{llYYYY}
    		\toprule
    		\textbf{Variable} & &\multicolumn{4}{c}{\textbf{PC}}\\
    		\cmidrule{3-6}
    		& & $PC_1$&$PC_2$&$PC_3$&$PC_1$\\
    		\midrule
             Explained     & Experiment & 49.40\% &45.09 \%& 4.69 \%&0.81\%\\
    		  variance     & Simulation & 40.17\% &33.23 \%& 26.51 \%&0.09\%\\		
    		\addlinespace
            Angle with    & Experiment & 87.68$^{\circ}$ & 90.1$^{\circ}$& 87.4$^{\circ}$& 3.48$^{\circ}$\\
    		$\operatorname{span}\{\mathbf{1_N}\}$                                & Simulation & 90.74$^{\circ}$ & 87.91$^{\circ}$& 91.19$^{\circ}$& 2.52$^{\circ}$
			\end{tabularx}
	\end{table}
	
	\begin{table}[!tbp]
		\centering
		\caption{Summary of UCM results between people forces and between finger forces}	\label{tab:UCM}
			\begin{tabularx}{\linewidth}{lYYY}
    		\toprule
    		& $V_{UCM}$&$V_{ORT}$& $\Delta V$ \\
    		\midrule
             Experiment & 2.2075  & 0.0020 & 1.3285 \\
    		 Simulation & 2.6945 & 0.0026 & 1.3282  \\
			\end{tabularx}
	\end{table}
	

	\section{Conclusion}
	This paper proposes a mathematical model for neuromotor control in collaborative behaviors using control theory. Inspired by our previous HHC experiment \cite{thesis}, we simulated two people's combined output performing a finger force matching task. Using Lyapunov analysis and the consensus protocol, we provided a control strategy for collaborative output-tracking, which was related to the UCM analysis and prove convergence to the desired transient and steady state behavior at the task and individual levels. We showed that the simulated results closely match the experiment. Our findings suggest a connection between consensus protocol and inter-personal motor synergy. The direct relationship is established between the agreement subspace (the task-level) and its complement (the UCM). The consensus dynamics are shown to not affect the task-level, providing a mechanism for mathematically describing how to handle the redundant degrees of freedom similar to how the CNS does based on the UCM theory.
	\balance
	\bibliographystyle{IEEEtran}
	\bibliography{root}
\end{document}